\documentclass{ecai}
\usepackage{graphicx}
\usepackage{latexsym}

\usepackage{amsmath}
\usepackage{amsthm}
\usepackage{booktabs}
\usepackage{algorithm}
\usepackage{algpseudocode}
\usepackage{mathtools}
\usepackage{amsfonts}
\usepackage{multirow}
\DeclareMathOperator*{\argmax}{arg\,max}

\newtheorem{lemma}[theorem]{Lemma}

\theoremstyle{definition}
\newtheorem{definition}{Definition}

\begin{document}

\begin{frontmatter}

\title{Safe Opponent Exploitation \\For Epsilon Equilibrium Strategies}

% \author{
%     Anonymous Submission
% }
\author[A]{\fnms{Linus}~\snm{Jeary}}
\author[B]{\fnms{Paolo}~\snm{Turrini}}
%\orcid{....-....-....-....}}
%\author[B]{\fnms{Third}~\snm{Author}\orcid{....-....-....-....}} % use of \orcid{} is optional

\address[A]{University of Warwick}
\address[B]{University of Warwick}

\begin{abstract}
In safe opponent exploitation players hope to exploit their opponents' potentially sub-optimal strategies while guaranteeing at least the value of the game in expectation for themselves.
Safe opponent exploitation algorithms have been successfully applied to small instances of two-player zero-sum imperfect information games, where Nash equilibrium strategies are typically known in advance. Current methods available to compute these strategies are however not scalable to desirable large domains of imperfect information such as No-Limit Texas Hold 'em (NLHE) poker, where successful agents rely on game abstractions in order to compute an equilibrium strategy approximation. This paper will extend the concept of safe opponent exploitation by introducing {\em prime-safe} opponent exploitation, in which we redefine the value of the game of a player to be the worst-case payoff their strategy could be susceptible to. This allows weaker epsilon equilibrium strategies to benefit from utilising a form of opponent exploitation with our revised value of the game, still allowing for a practical game-theoretical guaranteed lower-bound. We demonstrate the empirical advantages of our generalisation when applied to the main safe opponent exploitation algorithms.
\end{abstract}

\end{frontmatter}

\section{Introduction}
\noindent Following developments in AI for perfect information games such as chess \cite{campbell2002deep} and Go \cite{silver2016mastering}, AI for games with imperfect information such as No-Limit Texas Hold-em (NLHE) poker \cite{neumann1928theorie} was classified as the next milestone for AI's ability to beat the top players at their own game. %Games or scenarios of imperfect information refer to situations in which all players in a game are not able to observe all the information about the game states. A tactful demonstration of this phenomenon is found through the game of No-Limit Texas Hold-em (NLHE) poker \cite{neumann1928theorie}, since all opponent's cards are not observed until after a given round.

NLHE poker is an exceedingly large game; in the two-player heads-up (HU) setting alone it contains $10^{161}$ decision points in total \cite{johanson2013measuring}. For comparison, the game of chess has approximately $10^{43}$ \cite{shannon1950xxii}. In perfect information games, upon solving subgames in their decision trees, only the subtree representing the rest of the game needs to be considered. This notion is not promised in imperfect information games and significantly increases their complexity to solve \cite{brown2017safe}.

Significant progress has been made in the past decade, with Libratus being able to defeat top NLHE poker players in a one-on-one format using a game-theoretic approach \cite{brown2018superhuman}. Due to the two-player zero-sum environment Libratus was subject to, Nash equilibrium strategies are guaranteed to converge to a non-negative expected value in the long run regardless of the opponent's strategy \cite{nash1951non} \cite{neumann1928theorie}. The overarching aim for AIs like Libratus, therefore, is to try to find and play as close to a Nash equilibrium strategy as possible. 

Current methods to calculate Nash equilibrium strategies cannot be scaled to support the domain size of HU NLHE poker, hence abstraction techniques or algorithms must be applied before further computation. The abstraction techniques applied to NLHE poker consist of both information and action abstraction \cite{sandholm2015abstraction}. Information abstraction is exhibited by card abstraction, wherein approximately equivalent hands are clustered into \textit{buckets} to reduce decision points. Similarly to card abstraction, action abstraction is induced through clustering similar bet sizes together into buckets. After abstracting the game, an equilibrium strategy is computed and is subsequently mapped back to the full game resulting in an epsilon equilibrium strategy.

Playing accordingly to the computed epsilon equilibrium strategy will yield a desirable profit in expectation independent of the opponent being faced due to the approximate non-negative worst-case payoff. Whilst being a reliable foundation strategy, when playing against weaker opponents playing suboptimal strategies, the ability to effectively exploit their weaknesses would result in a significantly increased payoff. 

When attempting to aptly exploit an opponent in games like HU NLHE poker, we need to divert from playing an epsilon equilibrium strategy to try and extract value from our opponent's perceived weakness in their strategy without leaving ourselves open for exploitation. Previously presented opponent exploitation algorithms have aimed to achieve empirically sound performance as their main objective \cite{ganzfried2011game}. However, without the ability to theoretically evaluate and maintain the worst-case performance of these algorithms, we cannot ensure that the exploitative strategy being employed is not exploitable itself.
% Relying on only empirical performance in However, since we are granted a lot of game-theoretical privileges in the two-player zero-sum setting, we want to be able to ensure some theoretical guarantees in all our methods used within an AI. Additionally, we wish to refrain from relying on prior knowledge about our opponent or our exact domain of imperfect information. Ideally, we want a way to exploit our opponent without leaving ourselves open to exploitation during our periods of deviation.

We aim to extend the results produced from previous related work \cite{ganzfried2015safe} to allow epsilon equilibrium strategies to benefit from safe opponent exploitation. We then outline three amended algorithms to implement this generalisation and model their performances on simplified forms of poker.

\paragraph{Our Contribution.}

In this paper, we propose a generalisation of safe opponent exploitation algorithms which will allow agents playing a suboptimal epsilon-equilibrium strategy to gain additional utility compared to if they were to play according to their original static strategy. We do so by introducing prime-safety to ensure that an agent will not perform worse than the worst-case utility of the agent's initial static strategy when attempting to dynamically exploit the opponent. These algorithms are guaranteed to maintain this game-theoretic lower bound and will allow the exploitative player to safely punish the opponent's weaknesses throughout the game. This is a significant development due to the computation of Nash equilibria being a PPAD-hard problem \cite{daskalakis2009complexity} and approximate equilibrium strategies needing to be utilised for practical purposes in larger domains of imperfect information like HU NLHE. The techniques required in our algorithms can be scaled significantly higher than methods to compute exact optimal strategies for two-player zero-sum games and are consequently suitable for these larger domains \cite{brown2018depth}\cite{johanson2011accelerating}. The positive empirical performance of our algorithms in testing shows the potential for applying these opponent exploitation algorithms to improve the performance of agents in large imperfect information games.

Generalising prime-safety further to multiplayer games of imperfect information can be achieved similarly by reconstructing the value of the game away from the utility of playing a Nash equilibrium strategy. Although the exploitative strategy will have to be adapted, by similarly considering value as the value of the game to be the worst-case utility of our agent's static strategy,  similar algorithms guaranteeing safety could be developed for multiplayer zero-sum imperfect information games. The algorithm we describe allows any initial strategy to benefit from exploitation with a known lower bound. This is a valuable property as it remains an open question of what strategy should be deemed optimal even in very small domains of multiplayer zero-sum imperfect information games. As a result, we cannot select from a set of game-theoretically guaranteed optimal strategies in any way  \cite{ganzfried2018successful}.

Real-life implementations of safe opponent exploitation algorithms have previously suggested their application to cyber security games \cite{ganzfried2015safe}\cite{blythe2011testing}\cite{pita2012robust}. Our generalised algorithms allow for security game agents playing epsilon equilibrium strategies to capitalise on these weaknesses whilst maintaining a suitable known lower bound.

\paragraph{Paper Structure.} In Section \ref{sec:preliminaries} we present the key definitions and facts needed to present our algorithm, together with the variants considered in related literature. Section \ref{sec:prime} presents a theoretical analysis of prime-safety and Section \ref{sec:experiments} the experimental evaluation. We conclude with a few pointers to potential developments.

\section{Preliminaries}\label{sec:preliminaries}
\paragraph{Extensive-form games.}
Given a representative rooted game tree, a finite imperfect information extensive-form game is a game played with a finite set of $N$ players $\mathcal{P}$ and a set of histories $H$ including the initial empty history $\phi$ where the utilities for each player in $\mathcal{P}$ are determined by terminal states of the tree. $H$ consists of all the rooted terminal sequences of actions from a finite set $\mathcal{A}$ in the game tree, wherein each player $i$ is assigned to every strict subsequence in $H$. Moreover, each player's available actions accessible at a point in $h\in H$ from the game tree are denoted by $A(h)\in \mathcal{A}$. Additionally, there exists a Chance player $c$ representing chance events in the game where at history $h$, $c$ plays an action $a\in A(h)$ according to a fixed probability distribution known publicly by all players in $\mathcal{P}$. For each player $i$, their accessible nodes in the tree are split into information sets $I_i\in \mathcal{I}_i$ such that for all states in the same information set the player cannot distinguish between the sets and therefore must select an action with a constant distribution between all such states. Therefore, this action selection for a player $i$ at each accessible information set $I_i\in \mathcal{I}_i$ is formalised as $A(I_i)$. A pure strategy $s_i$ from the pure strategy space $S_i$ for a given player $i$ is a vector that directly determines the exact action player $i$ will take from $A(I_i)$, subject to $I_i$. Further, a mixed strategy $\sigma_i \in \Sigma_i$, where $\Sigma_i$ is the mixed strategy space for player $i$, is a function that at each assigns a probability distribution the player will employ at each $I_i$ over $A(I_i)$. A strategy profile $\sigma$ is a tuple of all player strategies $\sigma_i$ for $i=\{1,2,\ldots,N\}$ and $u_i(\sigma_i,\sigma_{-i})$ represents the utility function or expected payoff for player $i$ playing with respect to mixed strategy $\sigma_i$ against $\sigma_{-i}$. 
In applications such as poker, we are interested in examining finitely repeated games describing games repeated over finitely many rounds $T$, such that the total payoff for a player $i$ in the game will amount to the accumulated payoffs obtained over all rounds of play.

\paragraph{Equilibria.}
Let $\sigma_{-i}$ denote the mixed strategy profile of all other players aside from $i$ (For a two-player game, this will simply refer to the opponent strategy given we are playing as player $i$). For a given player $i$, their \emph{best response strategy} to $\sigma_{-i}$ is defined to be any strategy $\sigma_{i}$ satisfying, $\sigma_{i}=\underset{{\sigma{'}_i}\in{ \Sigma_i}}\argmax$ $u_i(\sigma{'}_i,\sigma_{-i})$.

The \emph{nemesis strategy} to a player $i$ is a pure strategy profile denoting a best response strategy to $\sigma_i$.

A strategy profile $\sigma$, where $\sigma=(\sigma_i,\sigma_{-i})$ is said to be a \emph{Nash equilibrium strategy} if $u_i(\sigma_i,\sigma_{-i})\geq u_i(\sigma{'}_i,\sigma_{-i})$ $\forall$ $\sigma{'}_i \in \Sigma_i$.

Similarly, given $\epsilon>0$, a strategy profile $\sigma$ is said to be an \emph{$\epsilon$-equilibrium strategy} if $u_i(\sigma_i,\sigma_{-i})\geq u_i(\sigma{'}_i,\sigma_{-i})-\epsilon$ $\forall$ $\sigma{'}_i \in \Sigma_i$.
Many Nash equilibria strategies may exist within a given game and Nash equilibrium strategies have been proven to exist in all non-cooperative games. Additionally, in two-player zero-sum games, playing a Nash equilibrium strategy guarantees that any player who plays the strategy will not lose in expected value independently of the opponent's strategy.

In two-player zero-sum games, John von Neumann's minimax theorem states that,

$\underset{\sigma_i\in\Sigma_i}\max\underset{\sigma_{-i}\in\Sigma_{-i}}\min u_i(\sigma_i,\sigma_{-i}) = \underset{\sigma_{-i}\in\Sigma_{-i}}\min\underset{\sigma_i\in\Sigma_i}\max u_i(\sigma_i,\sigma_{-i}).$

We can resultantly define the \emph{value of the game} to player $i$ to be denoted as $v_i$ to be the expected utility achieved by playing a minimax strategy.

\paragraph{Epsilon-safe Strategies.}
The exploitability of a strategy for player $i$ is defined as, $\text{expl}(\sigma_i)=v_i-\underset{{\sigma{'}_{-i}}\in{ \Sigma_{-i}}}\min u_i(\sigma_i,\sigma{'}_{-i})$.

Note that there always exists a pure best response strategy $s_{-i}\in S_{-i}$ to a mixed strategy $\sigma_{i}\in \Sigma_{i}$.

For a given player $i$, a strategy $\sigma_i \in \Sigma_i$ is referred to as being $\epsilon$-safe if the utility gained from playing the nemesis strategy $s_{-i}$ against $\sigma_i$ is within $\epsilon$ of $v_i$. Since the nemesis strategy is the worst-case strategy that maximises utility loss for player $i$, we can formally characterise a strategy $\sigma_i$ to be $\epsilon$-safe if, 
$
u_i(\sigma_i,s_{-i})\geq v_i-\epsilon, \forall s_{-i}\in S_{-i}, \text{for some } \epsilon \geq 0.
$
The set of strategies $\sigma_i\in \Sigma_i$ satisfying $\epsilon$-safety are defined to be of the set SAFE($\epsilon$). Note that the set of strategies within SAFE(0) represent the Nash equilibria strategies for the game, as the value of $v_i$ is guaranteed in expectation.
\paragraph{Counterfactual Regret Minimisation.}
Counterfactual Regret Minimisation (CFR) \cite{zinkevich2007regret} is a repeated tree traversal algorithm that, via $T$ iterations of self-play, will assign a level of \textit{regret} to each action in the decision tree. The CFR algorithm computes an updated strategy profile $\sigma^t$ at each iteration $t\in \{1\dots T\}$ and allows the average regret over all actions per information set to converge to zero.

CFR requires multiple full tree traversals to converge to an equilibrium strategy, consequently there exist multiple variants of counterfactual regret minimisation algorithms that aim to reduce the convergence rate and the required amount of the computation. Monte Carlo based sampling techniques can be utilised to reduce the number of tree traversals whilst ensuring that immediate counterfactual regrets are unaffected in expectation. Such algorithms are known as Monte Carlo counterfactual regret minimisation (MCCFR) algorithms. A common variant of MCCFR we use in this paper is external-sampling MCCFR \cite{lanctot2009monte} and is often used in large imperfect information games \cite{bowling2015heads}.

\section{Related Contributions}
McCracken and Bowling \cite{mccracken2004safe} introduce the notion of $\epsilon$-safe strategies/policies, which are strategies that guarantee a worst-case utility within $\epsilon$ of $v_i$ for a player $i$. They subsequently introduce a safe policy selection algorithm (SPS) for strategic form games. The algorithm allows the agent to play a dynamic $\epsilon^{(t)}$-safe strategy such that the value of $\epsilon^{(t)}$ is updated at each iteration of the game. $\epsilon^{(t)}$ is updated with respect to the opponent's model and their actions observed throughout the game by the rule $\epsilon^{(t)}$ is $\epsilon^{(t+1)}\leftarrow \epsilon^{(t)}+f(t+1)+u_i(\sigma^t_{-i},a^{(t)}_{-i})-v_i$
Where the function $f:\mathbb{N}\mapsto \mathbb{R}$ is a decay function subject to the constraints $f(t)>0$ $\forall$ $t\in\mathbb{N}$ and $\underset{T\rightarrow\infty}{\lim}\left(\frac{\sum^T_{t=1}f(t)}{T}\right)=0$.

The SPS algorithm was shown to perform experimentally well using the function assignment $f(t)=\frac{\beta}{t}$, however, the proof of safety for the algorithm is only ensured in the infinitely repeated setting (as $T\rightarrow\infty$). Since we are only interested in the finitely repeated domain, this is the key theoretical drawback of this paper.

Ganzfried and Sandholm \cite{ganzfried2011game} present an algorithm (Deviation-Based
Best Response DBBR), which attempts to exploit the opponent according to their opponent model, by playing a best response to the model under some conditions. DBBR can model opponents in imperfect information games, where private information is not necessarily observed at the end of each iteration of play. The main drawback of this work is due to the non-existence of a lower bound guarantee when switching to an exploitative strategy, therefore the reliance of the algorithm lies on the empirical performance over a small sample of opponent archetypes. This paper conjectured that it is impossible for an opponent exploitation algorithm to exist for a player $i$ without jeopardising the value of the game $v_i$ at each iteration of the game in expectation. This claim was disproved thereafter by the same authors in a 2015 paper which is central to the algorithms and methods used in this project.

Ganzfried and Sandholm \cite{ganzfried2015safe} introduce several safe opponent exploitation algorithms that guarantee that the agent will not risk losing the value of the game $v_i$ at each iteration of the game in expectation using methods from game theory. This guarantee encapsulates the notion of a \emph{safe} opponent exploitation algorithm. They propose a way to characterise situations for when safe opponent exploitation is possible. These situations are given to us when our opponent plays a strategy $\sigma_{-i}$ such that there exists an equilibrium strategy $\sigma^{*}_i \in \Sigma_i$ for player $i$ where $\sigma_{-i}$ is not a best response to $\sigma^{*}_i$.

These types of strategies are referred to as gift strategies and are the crux of the safe opponent exploitation algorithms proposed. The algorithms used a variable $k^t$ to represent the number of gifts player $i$ has been given by their opponent at iteration t ($1\leq t \leq T$).

Liu et. al \cite{liu2022safe} present a Safe Exploitation Search (SES) algorithm that uses real time search to attempt to exploit the opponent in subgames whilst maintaining safety with respect to the player's stage-game strategy. We want to focus on computing safe exploitation strategies in repeated games, where the notion of safety is distinct from the stage game environment.
%======================================================================

\section{Prime-Safe Exploitation Algorithms}\label{sec:prime}
%======================================================================
After abstracted settings of imperfect information games are processed through an algorithm like CFR, the calculated Nash equilibrium strategy will have to be translated in some way to the full game. After doing so, since some detail may have been lost about the full game in the abstraction, we will have derived an $\epsilon$-equilibrium strategy for the game. As a result, we are not able to apply the ``Best Equilibrium'' based safe opponent exploitation algorithms proposed by Ganzfried and Sandholm \cite{ganzfried2015safe} since we lose the ability to use a strategy in SAFE(0).

All the algorithms proposed in the previous paper require at least one exact Nash equilibrium strategy to be computed to maintain $v_i$ at each iteration of play. Since we want to have the ability to exploit an opponent with at least a single $\epsilon$-equilibrium strategy, we need to make the appropriate adaptations to the prior algorithms in order to achieve this relaxation. 

Due to these aforementioned details, we construct modifications of the three algorithms presented in \cite{ganzfried2015safe}, BEFEWP, BEFFE and RWYWE to attempt to successfully exploit opponents whilst conserving some form of game-theoretical guarantee concurrently.

% $\epsilon$-best response to be computed at each iteration of the game. $\epsilon$-best response strategies are computationally expensive to compute proportional to the size of the decision tree \cite{johanson2007computing}, thus we avoid using any algorithms requiring calculation of these strategies. 

\begin{definition}
Define $v_i{'}(\sigma_i)$ to be the expected utility of the best response strategy against player $i$, i.e. \begin{gather*}
v_i{'}(\sigma_i)=min_{\sigma_{-i}}u_i(\sigma_i,\sigma_{-i})
\end{gather*}
%$=v_i-\beta$ for $\beta = \text{expl}(\sigma_i)$.
\end{definition}

\begin{definition}
A strategy is said to be prime-safe in a repeated game if at each period of the game, the expected utility of a best response to the strategy guarantees at least a value of $v_i{'}(\sigma_i)$.
\end{definition}

\begin{definition}
An algorithm is said to be prime-safe in a repeated game if the strategy $\sigma^t_i$ maintained by the algorithm is prime-safe over $T$ steps.
\end{definition}

\begin{definition}
For a given player $i$, a strategy $\sigma_i \in \Sigma_i$ is referred to as being $\epsilon$-prime-safe if the utility gained from playing the nemesis strategy $s_{-i}$ against $\sigma_i$ is within $\epsilon$ of $v_i{'}(\sigma_i)$. Since the nemesis strategy is the worst-case strategy that maximises utility loss for player $i$, we can formally characterise a strategy $\sigma_i$ to be $\epsilon$-prime-safe if, 
\begin{gather*}
u_i(\sigma_i,s_{-i})\geq v_i{'}(\sigma_i)-\epsilon, \forall s_{-i}\in S_{-i}, \text{for some } \epsilon \geq 0.
\end{gather*}
The set of strategies $\sigma_i\in \Sigma_i$ satisfying $\epsilon$-prime-safety are defined to be of the set PSAFE($\epsilon$).
\end{definition}

\paragraph{Algorithm Description.}
The Epsilon Equilibrium and Full Exploitation When Possible (EEFEWP) algorithm is a generalisation of the aforementioned BEFEWP algorithm. The aim of the algorithm is to keep playing our epsilon-equilibrium strategy against our opponent until we can accumulate enough gifts to exploit by switching to a best response strategy with respect to our opponent model $M(\sigma_{-i})$ whilst ensuring prime-safety via the update rule for the variable $k^t$. 

To generalise the BEFEWP algorithm, we change the value of the game from $v_i$ to $v_i{'}(\sigma_i)$. This adaptation seems slight, however, some important distinctions and nuances arise as a result. We cannot assume we know of a collection of equilibria strategies like we can when we are playing according to $v_i$, since we may only calculate one $\epsilon$-equilibrium strategy for our game through a variant of CFR. Therefore, we cannot assume we are able to select a best equilibrium strategy according to the opponent model so instead we play any $PSAFE(0)$ strategy. Additionally, the exploitability of playing the best response strategy to the opponent model $M(\sigma_{-i})$ and the update rule for gift accumulation handled by $k^t$ are both calculated with respect to $v_i{'}(\sigma_i)$.
\begin{algorithm}%[!b]
\caption{EEFEWP for two-player zero-sum extensive-form games of imperfect information
where opponent’s private information is observed at the end of the game}\label{alg:cap4}
\begin{algorithmic}
\State \textbf{Input:} Strategy $\sigma_i^0$ for player $i$
\State $v_i{'}(\sigma_i) \gets min_{\sigma_{-i}}u_i(\sigma_i^0,\sigma_{-i})$
\State $k^1\gets0$
\For{$t = 1$ to $T$}
    \State $\sigma^t_{BR} \gets argmax_{\sigma_i}u_i(\sigma_i,M(\sigma_{-i})) $
    \State $\epsilon' \gets v_i{'}(\sigma_i) -min_{\sigma_{-i}}u_i(\sigma^t_{BR},\sigma_{-i})$
    \If{$\epsilon'\leq k^t$} 
        \State $\sigma^t_i \gets\sigma^t_{BR}$ %Play full BR
    \Else
        \State $\sigma^t_i \gets \sigma_i^0$
    \EndIf 
    \State Play action $a^t_i$ according to $\sigma^t_i$
    \State The opponent plays action $a^t_{-i}$ with observed private 
    \State information $\theta^t_{-i}$, according to the unobserved 
    \State distribution $\sigma^t_{-i}$
    \State Update $M$ with opponent's actions, $a^t_{-i}$, and his 
    \State private information $\theta^t_{-i}$
    \State $\tau^t_{-i} \gets$ strategy for the opponent that plays a best 
    \State response against $\sigma^t_i$ subject to the constraint that it 
    \State plays $a^t_{-i}$  on the path of play with private 
    \State information $\theta^t_{-i}$
    \State $k^{t+1}\gets k^t+u_i(\sigma^t_{i},\tau^t_{-i})-v_i{'}(\sigma_i)$ %kt - vstar + expected utility of playing strat_i against the opponents action
\EndFor
\end{algorithmic}
\end{algorithm}

It is not immediately clear how often agents playing noisier $\epsilon$-equilibria strategies will attempt exploitation and to what degree of success, thus the algorithm's merits will rely on positive empirical performance. Since we only switch to playing an exploitative strategy under unpredictable circumstances from the opponent's perspective, we can expect that rationally playing a best response strategy against an accurate opponent model will grant a strong gain in value in expectation. Moreover, the game-theory-based construction of EEFEWP still allows us to guarantee a lower bound in performance independently from our opponent. 

Now by updating the relevant lines in the algorithm we derive the Epsilon Equilibrium and Full Exploitation When Possible (EEFEWP) algorithm below. 

\begin{lemma}\label{thm:a}
If $\sigma_i^0$ is updated according to Extensive-Form EEFEWP, Then $k^t\geq0$ $\forall$ $t\geq1$.
\end{lemma}
\begin{proof}\ref{thm:a}. 
Since we have $k^1=0$, if $k^t\geq0$ for some $t\geq1$ then by construction of the EEFEWP algorithm, $\sigma_i^t$ has exploitability of at most $k^t$. Therefore,

$u_i(\sigma^t_i,\sigma^t_{-i})\geq v_i-\text{expl}(\sigma_i)-k^t = v_i{'}(\sigma_i)-k^t$.
\end{proof}
\begin{lemma}\label{thm:b}
EEFEWP is prime-safe.
\end{lemma}
\begin{proof}\ref{thm:b}.
Recall from the construction of the update rule for $k^t$ that,
\begin{align*}
k^{n+1}& = k^n+u_i(\sigma^n_{i},\tau^n_{-i})-v_i{'}(\sigma_i)\\
E[k^{n+1}]&=E[k^{n}+u_i(\sigma^{n}_i,\tau^{n}_{-i})-v_i{'}(\sigma_i)]\\
&= E[u_i(\sigma^{n}_i,\tau^{n}_{-i})]-E[v_i{'}(\sigma_i)]+E[k^{n}]\\
&\leq u_i(\sigma^{n}_i,\sigma^{n}_{-i})-v_i{'}(\sigma_i)+E[k^{n}]\\
&\leq\left[\sum^{n}_{t=1}\left(u_i(\sigma^t_i,\sigma^t_{-i})-v_i{'}(\sigma_i)\right)\right]+E[k^1]\\
&=\sum^{n}_{t=1}u_i(\sigma^t_i,\sigma^t_{-i})-nv_i{'}(\sigma_i)
\end{align*}
Therefore, by rearranging the above and letting $n=T$.
\begin{gather*}
\sum^{T}_{t=1}u_i(\sigma^t_i,\sigma^t_{-i})\geq E[k^{T+1}]+Tv_i{'}(\sigma_i).
\end{gather*}
Since $k^{T+1}\geq0 \implies E[k^{T+1}]\geq0$, then combining the above with Lemma 1,
\begin{gather*}
\sum^{T}_{t=1}u_i(\sigma^t_i,\sigma^t_{-i})\geq Tv_i{'}(\sigma_i).
\end{gather*}
Since $u_i(\sigma^t_i,\sigma^t_{-i})\geq v_i-\beta=v_i{'}(\sigma_i)$, where $\beta=\text{expl}(\sigma_i)$ EEFEWP is $\beta$-safe and equivalently prime-safe.
\end{proof}
Following the construction of EEFEWP, we now similarly adapt the BEFFE and RWYWE algorithm to fit our new constraints. We rename the modified versions of the algorithms to be EEFFE (Epsilon Equilibrium Followed by Full Exploitation) and PRWYWE (Prime Risk What You've Won in Expectation) respectively. 

The EEFFE algorithm is similar to EEFEWP but does not attempt to exploit the opponent when possible at each iteration but will attempt to fully exploit towards the end of the total number of rounds that will be played $T$.
\begin{algorithm}%[t]
\caption{EEFFE for two-player zero-sum extensive-form games of imperfect information
where opponent’s private information is observed at the end of the game}\label{alg:cap5}
\begin{algorithmic}
\State \textbf{Input:} Strategy $\sigma_i^0$ for player $i$
\State $v_i{'}(\sigma_i) \gets min_{\sigma_{-i}}u_i(\sigma_i^0,\sigma_{-i})$
\State $k^1\gets0$
\For{$t = 1$ to $T$}
    \State $\sigma^t_{BR} \gets argmax_{\sigma_i}u_i(\sigma_i,M(\sigma_{-i})) $
    \State $\epsilon' \gets v_i{'}(\sigma_i) -min_{\sigma_{-i}}u_i(\sigma^t_{BR},\sigma_{-i})$
    \If{$k^t \geq (T-t+1)(v_i{'}(\sigma_i)-\epsilon')$} 
        \State $\sigma^t_i \gets\sigma^t_{BR}$ %Play full BR
    \Else
        \State $\sigma^t_i \gets \sigma_i^0$
    \EndIf 
    \State Play action $a^t_i$ according to $\sigma^t_i$
    \State The opponent plays action $a^t_{-i}$ with observed private 
    \State information $\theta^t_{-i}$, according to the unobserved 
    \State distribution $\sigma^t_{-i}$
    \State Update $M$ with opponent's actions, $a^t_{-i}$, and his 
    \State private information $\theta^t_{-i}$
    \State $\tau^t_{-i} \gets$ strategy for the opponent that plays a best 
    \State response against $\sigma^t_i$ subject to the constraint that it 
    \State plays $a^t_{-i}$  on the path of play with private 
    \State information $\theta^t_{-i}$
    \State $k^{t+1}\gets k^t+u_i(\sigma^t_{i},\tau^t_{-i})-v_i{'}(\sigma_i)$ %kt - vstar + expected utility of playing strat_i against the opponents action
\EndFor
\end{algorithmic}
\end{algorithm}

\begin{lemma}\label{thm:c}
EEFFE is prime-safe.
\end{lemma}
\begin{proof}\ref{thm:c}.
The arguments follow similarly from the proof of Lemma 2.
\end{proof}

PRWYWE plays an $\epsilon$-best response at each iteration to the opponent model where $\epsilon= PSAFE(k^t)$. This allows the player to continually attempt to exploit at each iteration whilst having the aggression of the computed exploitative strategy directly constrained by the amount of gifts that have been accumulated over previous iterations through the update rule for $k^t$.

\begin{algorithm}%[!t]
\caption{PRWYWE for two-player zero-sum extensive-form games of imperfect information
where opponent’s private information is observed at the end of the game}\label{alg:cap6}
\begin{algorithmic}
\State \textbf{Input:} Strategy $\sigma_i^0$ for player $i$
\State $v_i{'}(\sigma_i) \gets min_{\sigma_{-i}}u_i(\sigma_i^0,\sigma_{-i})$
\State $k^1\gets0$
\For{$t = 1$ to $T$}
    \State $\sigma^t_{i} \gets argmax_{\sigma_i\in PSAFE(k^t)}u_i(\sigma_i,M(\sigma_{-i}))$
    \State Play action $a^t_i$ according to $\sigma^t_i$
    \State The opponent plays action $a^t_{-i}$ with observed private 
    \State information $\theta^t_{-i}$, according to the unobserved 
    \State distribution $\sigma^t_{-i}$
    \State Update $M$ with opponent's actions, $a^t_{-i}$, and his 
    \State private information $\theta^t_{-i}$
    \State $\tau^t_{-i} \gets$ strategy for the opponent that plays a best 
    \State response against $\sigma^t_i$ subject to the constraint that it 
    \State plays $a^t_{-i}$  on the path of play with private 
    \State information $\theta^t_{-i}$
    \State $k^{t+1}\gets k^t+u_i(\sigma^t_{i},\tau^t_{-i})-v_i{'}(\sigma_i)$ %kt - vstar + expected utility of playing strat_i against the opponents action
\EndFor
\end{algorithmic}
\end{algorithm}
\begin{lemma}\label{thm:d}
PRWYWE is prime-safe.
\end{lemma}
\begin{proof}\ref{thm:d}.
The arguments follow similarly from the proof of Lemma 2 and from the proof of RWYWE in \cite{ganzfried2015safe}.
\end{proof}

These algorithms allow any valid strategy to attempt to exploit their opponent no matter what degree of initial exploitability the strategy presents, while still guaranteeing a game-theoretic lower-bound on the utility to be gained by doing so. The algorithms are directly applicable to agents in two-player zero-sum imperfect information games with some $\epsilon$-equilibria due to using CFR over finitely many iterations or by translating Nash equilibria strategies calculated for abstracted instances of a game to the full game. Both of these scenarios are very commonplace, especially when operating in much larger domains so the application of the modified algorithms to an agent's strategy may be beneficial if empirical results show potential. 

We wish to demonstrate EEFEWP, EEFFE and PRWYWE's performance specifically within small abstracted settings of imperfect information to test the limits of the algorithm's capabilities.

\section{Experiments}\label{sec:experiments}
%======================================================================
In line with previous work demonstrating the performance of exploitation algorithms within HU NLHE poker \cite{hoehn2005effective}\cite{ganzfried2015safe}, we test the performance of the EEFEWP, EEFFE and PRWYWE algorithms on variants of simplified Kuhn poker to evaluate each algorithm's potential empirically. 

We want to specifically establish variants of poker, such that we can abstract the game to some degree with regard to action and card abstraction. Instead of using a larger game such as Leduc poker \cite{southey2012bayes}, we aim to keep the size of the game as small as possible to allow us to conduct multiple experiments with a lower measured variance.

\paragraph{Kuhn poker}
Kuhn poker is a very simple setting of poker, used to model various techniques in imperfect, two-player zero-sum games \cite{kuhn1950simplified}.

Kuhn poker is a repeated game played with two players, P1 and P2, and a deck of 3 cards $c_1,c_2,c_3$ such that $c_1>c_2>c_3$, where the detailed rules of the game are as follows:
\begin{itemize}
    \item[$1)$] At the start of an iteration of the game, each player initially bets £1 into the pot (antes £1).
    \item[$2)$] Then, both players are dealt 1 card each at random from a deck of 3 cards and the last card remains unseen. Note that all cards aside from the player's own card are unobserved at this point. 
    \item[$3)$] P1 then chooses to bet £1 into the pot or pass the action on to P2 (check to P2).
        \begin{itemize}
            \item[$-$] If P1 chooses to bet £1, P2 can either match the bet of £1 submitted to the pot (call) or forfeit their hand by folding.
            \begin{itemize}
                \item[$-$] If P2 chooses to fold. P1 wins the pot.
            \end{itemize}
            \item[$-$] If P1 chooses to check, P2 can either bet £1 towards the pot or pass on making any bet (check).
                \begin{itemize}
                    \item[$-$] If P2 chooses to bet £1, P1 can either match the bet of £1 submitted to the pot (call) or forfeit their hand by folding.
                    \begin{itemize}
                        \item[$-$] If P1 chooses to fold. P2 wins the pot.
                    \end{itemize}
                \end{itemize}
        \end{itemize}
    \item[$4)$] Finally, if neither P1 or P2 has folded, both players make their private cards observable and the larger card wins the pot.
\end{itemize}
% \begin{figure}[h]
% \centering
% \includegraphics[scale=0.42]{LaTeX/equilibriatable.png}
% \caption{Table representing the full Nash equilibria strategy family paramaterised by $\alpha \in [0,1/3]$. The first row represents the non-terminal public history sequences, where B=Bet and K=Check and the first column represents the three possible cards in the deck. Green columns represent decision points for P1 and red columns represent decision points for P2. The (Row, Column) or (Card, History) pair represents the information set at a specific game state and the respective entries represent the probability distribution of playing a betting or passing action respectively.} 
% \end{figure}

The value of the game $v_1$ for P1 is -1/18 and due to the notion of operating in a two-player zero-sum environment, $v_2$ for P2 is +1/18. Intuitively, the negative value of the game for P1 is due to P2 having the advantage of being able to select an action having observed P1's action. This observation of P1's first action can often imply at least some information about the theoretical strength of P1's private card. Therefore, a negative average payoff over a number of games for P1 is not directly indicative of a weak strategy.
\paragraph{Generalised Kuhn Poker.}
% Kuhn poker is a very simple setting of poker, used to model various techniques in imperfect, two-player zero-sum games \cite{kuhn1950simplified}.

Since we want to examine the behaviour of our algorithm on abstracted settings, we wish to generalise the original game of Kuhn poker to be able to support abstractions. We examine the action and card abstraction settings separately to refine our analysis of the algorithm's performance. For the card-abstraction setting, we generalise vanilla Kuhn poker to accept a deck of 6 cards A, K, Q, J, T and 9 and translate $\epsilon$-equilibria strategies from lower dimension games of 3, 4 and 5 cards to the full game. Similarly, for the action-abstraction environment, we generalise the original Kuhn poker game by allowing 4 different bet sizes £1, £2, £3 and £4 and translate $\epsilon$-equilibria strategies from lower dimension games of 1, 2 and 3 bet sizes to the full game.

After we compute these strategies, we test them against four classes of opponents and evaluate the performance of the abstracted agents before and after using the exploitation algorithms.

\paragraph{Abstraction Method.}
Since the accuracy of our abstracted strategies is not assumed to be high by our algorithm we would ideally demonstrate an improvement in performance for all static strategies. Consequently, we simply translate each of the smaller settings of the games to the full game using a suitable uniform mapping. More sophisticated mappings have been previously established that aim to reduce the exploitability of abstracted game strategies' performance post-translation to the full game \cite{ganzfried2013action}. 

Since we are working within a very small domain, we expect to see pretty large amounts of exploitability between abstraction levels. It also is not really of much concern how exploitable each layer of abstraction is, since ideally, we would like our algorithm to improve the performance of any static strategy. As a result, we simply implement a uniform mapping for both card and action abstraction environments.
\begin{table*}%[ht]
    \centering
    \begin{tabular}{ccccccc}
        \toprule
        \multicolumn{7}{c}{6-card Kuhn Poker experiments} \\ \midrule
        &Abstraction & $v_1{'}$ & Random & Sophisticated & Dynamic & Equilibrium \\ \midrule
        \multirow{4}{*}{EEFFE} & None & -0.0611 & 0.1717 ± 0.0004 & -0.0370 ± 0.0005 & -0.0420 ± 0.0004 & -0.0611 ± 0.0005 \\
        & 5 cards & -0.0721 & 0.0786 ± 0.0005 & -0.0561 ± 0.0004 & -0.0607 ± 0.0004 & -0.0688 ± 0.0005 \\
        & 4 cards & -0.0750 & 0.1155 ± 0.0005 & -0.0422 ± 0.0005 & -0.0607 ± 0.0004 & -0.0610 ± 0.0004 \\
        & 3 cards & -0.0978 & 0.0702 ± 0.0004 & -0.0541 ± 0.0004 & -0.0851 ± 0.0004 & -0.0707 ± 0.0005 \\ \midrule
        \multirow{4}{*}{EEFEWP} & None & -0.0611 & 0.2017 ± 0.0006 & -0.0366 ± 0.0005 & -0.0397 ± 0.0005 & -0.0611 ± 0.0005 \\
        & 5 cards & -0.0721 & 0.1172 ± 0.0008 & -0.0555 ± 0.0005 & -0.0581 ± 0.0005 & -0.0679 ± 0.0005 \\
        & 4 cards & -0.0750 & 0.1419 ± 0.0008 & -0.0418 ± 0.0005 & -0.0577 ± 0.0005 & -0.0611 ± 0.0005 \\
        & 3 cards & -0.0978 & 0.0973 ± 0.0009 & -0.0533 ± 0.0004 & -0.0823 ± 0.0004 & -0.0695 ± 0.0004 \\ \midrule
        \multirow{4}{*}{PRWYWE} & None & -0.0611 & 0.2701 ± 0.0006 & -0.0360 ± 0.0005 & -0.0352 ± 0.0005 & -0.0609 ± 0.0005 \\
        & 5 cards & -0.0721 & 0.3710 ± 0.0007 & -0.0298 ± 0.0005 & -0.0347 ± 0.0005 & -0.0612 ± 0.0005 \\
        & 4 cards & -0.0750 & 0.4004 ± 0.0007 & -0.0266 ± 0.0005 & -0.0339 ± 0.0005 & -0.0619 ± 0.0004 \\
        & 3 cards & -0.0978 & 0.4813 ± 0.0005 & -0.0192 ± 0.0005 & -0.0472 ± 0.0005 & -0.0621 ± 0.0005 \\ \midrule
        \multirow{4}{*}{Equilibrium} & None & -0.0611 & 0.1407 ± 0.0004 & -0.0376 ± 0.0005 & -0.0412 ± 0.0004 & -0.0611 ± 0.0005 \\
        & 5 cards & -0.0721 & 0.0501 ± 0.0004 & -0.0568 ± 0.0005 & -0.0596 ± 0.0004 & -0.0692 ± 0.0004 \\
        & 4 cards & -0.0750 & 0.0791 ± 0.0004 & -0.0424 ± 0.0004 & -0.0597 ± 0.0004 & -0.0612 ± 0.0004 \\
        & 3 cards & -0.0978 & 0.0356 ± 0.0004 & -0.0553 ± 0.0004 & -0.0845 ± 0.0004 & -0.0706 ± 0.0004 \\ \bottomrule
        % \multicolumn{2}{l}{Best Response} & 0.3001 ± 0.0006 & -0.0204 ± 0.0004 & -0.0205 ± 0.0004 & -0.0555 ± 0.0004 \\ \bottomrule
    \end{tabular}
    \caption{For each algorithm, the agent plays subject to four different degrees of abstracted strategies within the 6-card Kuhn Poker variant. All agent's performances are measured by the total payoff/total hands played against an opponent. Each normalised payoff includes a ± to be the 95\% confidence interval established after all hands are played against the relevant opponent.}
\end{table*}
\begin{table*}%[!ht]
    \centering
    \begin{tabular}{ccccccc}
        \toprule
        \multicolumn{7}{c}{4-bet size Kuhn Poker experiments} \\ \midrule
        & Abstraction & $v_1{'}$ & Random & Sophisticated & Dynamic & Equilibrium \\ \midrule
        \multirow{4}{*}{EEFFE} & None & -0.0556 & 0.3206 ± 0.0006 & -0.0265 ± 0.0004 & -0.0287 ± 0.0004 & -0.0556 ± 0.0004 \\
        & 3 bet sizes & -0.1178 & 0.3973 ± 0.0006 & -0.0219 ± 0.0004 & -0.0861 ± 0.0005 & -0.0556 ± 0.0004 \\
        & 2 bet sizes & -0.1789 & 0.5029 ± 0.0007 & -0.0243 ± 0.0004 & -0.1298 ± 0.0005 & -0.0613 ± 0.0004 \\
        & 1 bet size & -0.2972 & 0.5802 ± 0.0011 & -0.0320 ± 0.0004 & -0.2376 ± 0.0005 & -0.0707 ± 0.0004 \\ \midrule
        \multirow{4}{*}{EEFEWP} & None & -0.0556 & 0.4453 ± 0.0007 & -0.0175 ± 0.0004 & -0.0071 ± 0.0005 & -0.0551 ± 0.0004 \\
        & 3 bet sizes & -0.1178 & 0.4472 ± 0.0008 & -0.0159 ± 0.0004 & -0.0637 ± 0.0005 & -0.0553 ± 0.0004 \\
        & 2 bet sizes & -0.1789 & 0.5702 ± 0.0009 & -0.0188 ± 0.0004 & -0.1115 ± 0.0005 & -0.0570 ± 0.0004 \\
        & 1 bet size & -0.2972 & 0.6483 ± 0.0008 & -0.0157 ± 0.0004 & -0.2143 ± 0.0005 & -0.0682 ± 0.0004 \\ \midrule
        \multirow{4}{*}{PRWYWE} & None & -0.0556 & 0.5826 ± 0.0009 & -0.0097 ± 0.0004 & 0.0024 ± 0.0004 & -0.0556 ± 0.0004 \\
        & 3 bet sizes & -0.1178 & 0.7701 ± 0.0017 & -0.0026 ± 0.0004 & -0.0428 ± 0.0004 & -0.0556 ± 0.0004 \\
        & 2 bet sizes & -0.1789 & 0.6514 ± 0.0011 & 0.0013 ± 0.0004 & -0.0969 ± 0.0004 & -0.0574 ± 0.0004 \\
        & 1 bet size & -0.2972 & 0.6622 ± 0.0011 & -0.0073 ± 0.0004 & -0.1547 ± 0.0005 & -0.0571 ± 0.0004 \\ \midrule
        \multirow{4}{*}{Equilibrium} & None & -0.0556 & 0.2389 ± 0.0006 & -0.0269 ± 0.0004 & -0.0176 ± 0.0004 & -0.0556 ± 0.0004 \\
        & 3 bet sizes & -0.1178 & 0.2671 ± 0.0006 & -0.0232 ± 0.0004 & -0.0791 ± 0.0005 & -0.0556 ± 0.0004 \\
        & 2 bet sizes & -0.1789 & 0.3809 ± 0.0007 & -0.0303 ± 0.0004 & -0.1217 ± 0.0005 & -0.0617 ± 0.0004 \\
        & 1 bet size & -0.2972 & 0.4323 ± 0.0007 & -0.0420 ± 0.0004 & -0.2232 ± 0.0006 & -0.0766 ± 0.0004 \\ \bottomrule
        % \multicolumn{2}{l}{Best Response} & 0.3001 ± 0.0006 & -0.0204 ± 0.0004 & -0.0205 ± 0.0004 & -0.0556 ± 0.0004 \\ \bottomrule
    \end{tabular}
    \caption{The agent player plays subject to four different degrees of abstracted strategies within the 4-bet size Kuhn Poker variant. Otherwise the same details apply as outlined in Table 1.}
\end{table*}
\paragraph{Opponent Classes.}
After developing our two test environments, we need various opponent archetypes to play against to explore the algorithms' potential. These opponent classes directly follow those constructed in the 2015 paper on Safe Opponent Exploitation \cite{ganzfried2015safe} as we are directly adapting from one of the paper's algorithms. The four types of opponents are:
\begin{itemize}
    \item Random: This opponent chooses an available action at every decision point uniformly at random.
    \item Dynamic: Opponent chooses an available action at every decision point uniformly at random for the first 100 hands and then switches to playing a full best response strategy for the last 900.
    \item Equilibrium: Opponent plays a low exploitability $\epsilon$-equilibrium strategy calculated via 100,000,000 iterations of MCCFR.
    \item Sophisticated: At each decision point, the opponent plays within $p$ $=$ 0.2 of a low exploitability $\epsilon$-equilibrium strategy calculated via MCCFR.
\end{itemize}

\paragraph{Algorithm Initialisation.}
In testing, we play our agent as P1 for 1,000 hands against each of the four opponent classes in both card and action abstraction settings, with and without each exploitation algorithm over 40,000 repetitions. P1's base static $\epsilon$-equilibrium strategy is pre-computed for some small $\epsilon>0$ via 100,000,000 iterations of MCCFR in each abstracted setting. We want our choice of modelling algorithm to fairly represent the opponent, thus we use a model similar to the one used in the DBBR algorithm \cite{ganzfried2011game}. 

We make a slight modification to the Dirichlet Prior distribution in DBBR, instead of initialising the action probabilities from a Nash equilibrium strategy, we use the $\epsilon$-equilibria strategies obtained after translating from the abstracted settings to the full games. We define our opponent model $M$ with respect to the modified Dirichlet Prior distribution initialised by $5$-prior hands as this produced favourable empirical performance.

%\paragraph{Results and Evaluation.}

%Retest EEFEWP 2 bet equil
\paragraph{Results Table.}
After computing the normalised payoff of all layers of abstraction in both environments against the four opponent classes both with and without the exploitation algorithms, we can begin to evaluate the success of employing each of the different algorithms in these different scenarios. We measure each entry as an average payoff from playing 1,000 hands over 40,000 iterations subject to a 95\% confidence interval to ensure statistical significance in our findings. Similarly to vanilla Kuhn poker, since we play as player 1, the value of the game $v_1$ is negative and is not representative of a poor payoff.

\paragraph{Results Analysis and Evaluation.}

Against all opponent classes and for all abstractions, both the multiple cards and bet size environments benefit significantly from implementing both EEFEWP and PRWYWE.

Despite a lower exploitability, the 4-card abstraction outperforms the 5-card abstraction against all opponent classes before and after using the exploitation algorithms aside from a marginal difference against the dynamic opponent without exploitation. This is to be expected since if a strategy has a lower worst-case utility than another strategy, it is not the case that that strategy is guaranteed to obtain a lower utility against all other possible strategies in the space versus a strategy with a higher worst-case utility. Furthermore, the number of possible opponent strategies that can maximally exploit a given strategy may be significantly fewer than a less exploitable strategy, but due to our worst-case analysis a strategy's strength is not necessarily implied empirically.

Similarly, all agents in the multiple bet size setting are seen to benefit significantly from the usage of the EEFEWP and PRWYWE algorithm. Due to the increased amount of bet sizes, the bet abstracted agents are much more exploitable than the abstracted agents in the 6-card environment. The more frugal betting strategy of less exploitable agents allows for a more consistent spread of performance over all opponent classes in contrast to the weaker agent's higher payoff mainly against the random opponent class. Additionally, the more exploitable agents are able to consistently reduce the gap between playing their original inferior strategy and performance from playing a Nash equilibrium strategy as a result of prime-safe exploitation.

The EEFFE algorithm underperforms when compared to the non-exploiting player playing against the dynamic opponent. This weak performance was similarly demonstrated by the original BEFFE algorithm's performance. Although $v_1{'}$ is still maintained by EEFFE, the empirical performance of the algorithm is completely outclassed by the other two exploitation algorithms.

Overall, both EEFEWP and PRWYWE produce a very promising set of results over both abstraction settings. Significant gain is established by every agent playing strategies from all levels of abstraction via prime-safe exploitation and as expected, the prime value of the game $v_1{'}$ is never jeopardised in any case.

\section{Conclusion} 
We have shown that opponent exploitation algorithms maintaining a sensible lower bound can be successfully applied to agents playing $\epsilon$-equilibria strategies in extensive-form games of imperfect information. We introduced three algorithms for this setting by generalising the algorithms presented in previous work \cite{ganzfried2015safe} to operate in the case where our agent's strategy may be approximate and formally proved its lower bound guarantee. This lower bound was proved as a result of introducing the concept of prime-safety, which guarantees the worst-case utility of playing an $\epsilon$-equilibrium strategy. The lower bound is significant due to opponent exploitation algorithms requiring us to play a dynamic strategy, where we could theoretically perform much worse than our lower bound.

Upon testing over two different types of abstraction, we showed that PRWYWE and EEFEWP display strong performances across the board. This strength held over a diverse selection of opponent archetypes including dynamic opponents, who played an unrealistically exploitative strategy for 90\% of the 1000 hands played.

Practically, the EEFEWP algorithm may be preferable to PRWYWE as it does not rely on an $\epsilon$-best response to be computed at each iteration of the game. $\epsilon$-best responses  are expensive to compute relative to the size of the decision tree \cite{johanson2007computing} in comparison to the standard best-response computation required in PRWYWE. 

The success of our generalised algorithms can be extended to other games to allow for a variety of strategies to benefit from exploitation. For much larger domains like HU NLHE poker, PRWYWE's computational requirements are feasible \cite{johanson2011accelerating}, however since AI built for solving HU NLHE are often using their real-time decision-making to compute methods like subgame solving, additionally implementing a prime-safe opponent exploitation algorithm would need to be admissible by the hardware.

\bibliography{ecai}
\end{document}